\documentclass[12pt,reqno]{amsart}
\usepackage{amsmath,amssymb}
\usepackage{hyperref}
\usepackage{graphicx}
\usepackage{booktabs,tabularx,array}
\usepackage{dsfont} 
\usepackage[margin=1in]{geometry}
\newtheorem{theorem}{Theorem}[section]
\newtheorem{lemma}[theorem]{Lemma}
\newtheorem{proposition}[theorem]{Proposition}
\newtheorem{definition}[theorem]{Definition}
\newtheorem{example}[theorem]{Example}
\newtheorem{remark}[theorem]{Remark}
\newtheorem{corollary}[theorem]{Corollary}

\DeclareMathOperator{\id}{id}
\DeclareMathOperator{\tr}{tr}

\title[Entropy Contraction of Generalized Quantum Depolarization]
{Tight Entropy Contraction of Generalized Quantum Depolarization}

\author{Li Gao}
\address{School of Mathematics and Statistics, Wuhan University, Wuhan 430072, China}
\address{Wuhan Institute of Quantum Technology, Wuhan 430075, China}
\thanks{School of Mathematics and Statistics, Wuhan University, Wuhan 430072, China}
\thanks{Wuhan Institute of Quantum Technology, Wuhan 430075, China}
\email{gao.li@whu.edu.cn}
\author{Long Zhao}
\address{School of Mathematics and Statistics, Wuhan University, Wuhan 430072, China}
\email{zhaolong@whu.edu.cn}
\date{\today}

\begin{document}

\begin{abstract}We establish upper and lower bounds for 
relative entropy contraction of generalized quantum depolarizing channels and semigroups.
Our bound provides tight first order asymptotic of the contraction rate in terms of the dimension constant. 
One side estimate are based on sharp reverse ratio and convexity of relative entropy of two states, which can be derived from the recently introduced
Hockey-Stick quantum $f$-divergence, and also independently, Bogoliubov--Kubo--Mori quantum Fisher information metric. The other side follows from the existence of index achieving pure state with respect to a general conditional expectation. Our results extend to the complete entropy contraction rate, tensor stable estimates for product dynamics.
Examples include quantum depolarization, dephasing, and
compact group symmetrization, with consequences for the decay rate of coherence and
asymmetry.
\end{abstract}

\maketitle

\everymath{\displaystyle}
\section{Introduction}
\label{sec:introduction}

Relative entropy is a fundamental measure in both classical and
quantum information  theory. Its classical form originates in Shannon's information
theory~\cite{Shannon1948} and is known as the Kullback--Leibler
divergence~\cite{KullbackLeibler1951}. Its quantum version was introduced
by Umegaki~\cite{Umegaki1962} and extended to general von Neumann algebras
by Araki~\cite{Araki1976}. In quantum information theory, relative entropy
plays a central role in quantifying distinguishability between quantum states, with wide applications
ranging from hypothesis testing to quantum information processing
~\cite{OhyaPetz1993,Ruskai2002,Vedral2002}. For two quantum states
\(\rho\) and \(\sigma\) with compatible support, it is defined as
\[
D(\rho\|\sigma)=\tr[\rho(\ln\rho-\ln\sigma)].
\]
A fundamental property underlying
its wide applications is the data-processing inequality: for any quantum channel
\(\Phi\),
\[
D(\Phi(\rho)\|\Phi(\sigma))
\le D(\rho\|\sigma).
\]

The data-processing inequality is qualitative: it describes the 
distinguishability is monotone non-increasing but does not quantify the rate of entropy contraction.
Quantitative refinements, such as strong data-processing inequalities (SDPIs)
and modified logarithmic Sobolev inequalities (MLSIs), provide explicit
contraction rates and characterize convergence of Markovian dynamics
\cite{DiaconisSaloffCoste1996,Raginsky2016}. In the quantum setting,
connections between logarithmic Sobolev inequalities and other functional
inequalities, such as hypercontractivity, Poincar\'e inequalities, and
transport cost inequalities, have been extensively developed for quantum
Markov semigroups~\cite{OlkiewiczZegarlinski1999,KastoryanoTemme2013,
DattaRouze18,RouzeDatta2019,carlen2017gradient,CarlenMaas18}. Related
inequalities for non-primitive quantum Markov semigroups and complete
formulations for quantum channels have been developed in
\cite{BardetRouze2022,GaoRouze2022}. However, obtaining explicit and
dimension independent contraction estimates for general quantum dynamics,
especially in the presence of nontrivial fixed-point algebras, remains
challenging.

A staring point of this work is the following sharp ratio bounds for relative entropy.
\begin{proposition}[Entropy ratio bounds]
\label{thm:intro-ratios}
Let \(\rho,\sigma\) be two quantum states, let \(0<p<1\), and set
\[
 C:=\inf\{c>0:\rho\le c\sigma\}>1,
 \qquad
 \phi(x):=x\ln x-x+1.
\]
Then
\begin{align}
 \frac{D(\sigma\|\rho)}{D(\rho\|\sigma)}
 &\ge \frac{C-1-\ln C}{C\ln C-C+1},
 \label{intro-reverse-ratio}\\
 \frac{\phi\bigl(1+p(C-1)\bigr)}{\phi(C)}
 &\ge
 \frac{D\bigl(p\rho+(1-p)\sigma\big\|\sigma\bigr)}
 {D(\rho\|\sigma)}
 \ge \phi(1-p).
 \label{intro-convexity-ratio}
\end{align}
\end{proposition}

The first inequality gives a universal reverse relative entropy estimate,
while the second provides two-sided bounds for the entropy contraction along
the linear interpolation \(\rho_p=p\rho+(1-p)\sigma\). These bounds essentially follows from the comparison theorem of Hockey-Stick quantum $f$-divergence \cite[Theorem 4.2]{BeigiHircheTomamichel2025} with explicit analysis on functions parameter corresponding the relative entropy, its reverse, and also convex combination variants. We also provide an independent argument using
Bogoliubov--Kubo--Mori metric and Chebyshev's integral inequality. The remarkable feature is
that all bounds are determined only by the minimal order parameter
\(C\), or equivalently, the max relative entropy $D_{\max}(\rho||\sigma)=\ln C$ \cite{datta2009min}. All the bounds above are optimal given $C$, as they generalize the optimal comparison of classical $f$-divergence by \cite{SasonVerdu2016}.

We apply the above entropy ratios to study the entropy contraction rates of generalized quantum depolarizing process given by faithful conditional expectations.
Namely, for am idempotent quantum channel $E=E^2$ with a faithful invariant state,
we consider the following analog of depolarizing
semigroup \(P_t^E\) and channel \(\Phi_p^E\):
\[
 P_t^E=e^{-t}\id+(1-e^{-t})E,
 \qquad
 \Phi_p^E=p\,\id+(1-p)E.
\]
 This definition includes the standard
depolarizing and dephasing cases.  For depolarization, \(E\) is the replacer
channel \(E(\rho)=\tr(\rho)\sigma\); for dephasing, it is the pinching map
\(E(\rho)=\sum_i |i\rangle\langle i|\rho|i\rangle\langle i|\).
An initial state \(\rho\) always compared to its projection \(E(\rho)\).
The max-relative entropy of this pair is controlled by the following
structural constant of \(E\):
\[
 C(E):=\inf\{C>0:\rho\le CE(\rho)\text{ for every state }\rho\}.
\]
Here \(\leq\) means the positivity order.
This order constant is called the Pimsner--Popa index
introduced in~\cite{PimsnerPopa1986}.  Its completely bounded amplification
\[
 C_{\mathrm{cb}}(E):=\sup_{k\ge1}C(E\otimes\id_{M_k})
\]
plays the corresponding role for complete entropy inequalities.  These index
quantities provide a dimension measure of how far an initial state $\rho$ can be away from its
 conditional expectation $E(\rho)$ in the fixed-point space.
  For
\(C>1\), we set two functions
\[
 \mathsf A(C):=
 \frac12\left(1+
 \frac{C-1-\ln C}{C\ln C-C+1}\right),
 \qquad
 \mathsf B_p(C):=
 \frac{\phi(1+p(C-1))}{\phi(C)}.
\]
At \(C=1\), these functions are defined by continuity:
\(\mathsf A(1)=1\) and \(\mathsf B_p(1)=p^2\).
Define \(\eta(\Phi_p^E)\) and \(\alpha(P_t^E)\) as the optimal constants such
that
\[
\begin{aligned}
 D(\Phi_p^E(\rho)\|E(\rho))
 &\le \eta(\Phi_p^E)D(\rho\|E(\rho)),\\
 D(P_t^E(\rho)\|E(\rho))
 &\le e^{-2\alpha(P_t^E)t}D(\rho\|E(\rho)).
\end{aligned}
\]
Combining Proposition 1.1 with estimate from the other side, we obtain the tight first order
asymptotic for the entropy contraction rates. Moreover, applying these two-sided bounds 
to the corresponding entropy-contraction functionals yields the following parallel bounds.

\begin{theorem}[Tight entropy contraction ratio]
\label{cor:intro-entropy-contraction-bounds}
For the above dynamics, 
\[
 \begin{aligned}
 A(C(E))\le \alpha(P_t^E)&\le \mathsf U(C),\\
 L_p(C)\le \eta(\Phi_p^E)&\le\mathsf B_p(C(E)),
 \end{aligned}
\]In particular, we have 
\begin{align}
 \alpha(P_t^E)
  &=\frac12+\frac{1}{2\ln C}
    +O\!\left(\frac{\ln\ln C}{(\ln C)^2}\right),\\
  \eta(\Phi_p^E)
  &=p+\frac{p\ln p}{\ln C}
    +O\!\left(\frac{\ln\ln C}{(\ln C)^2}\right).
\end{align}
\end{theorem}

Here
the matching first order estimate are given by the optimization expressions 
that for $C>1$ and $0<p<1$, set
\begin{equation}
\begin{aligned}
 \mathsf U(C)
 &:=\frac12\left(1+
 \inf_{\substack{0<x<1\\x\ne C^{-1}}}
 \frac{D_2(C^{-1}\|x)}{D_2(x\|C^{-1})}\right),
 \nonumber\\
 \mathsf L_p(C)
 &:=\sup_{\substack{0<x<1\\x\ne C^{-1}}}
 \frac{D_2(px+(1-p)C^{-1}\|C^{-1})}
      {D_2(x\|C^{-1})}.
\end{aligned}
\label{eq:gd-binary-bound-functions}
\end{equation}
where $D_2$ is the binary relative entropy 
\[
 D_2(u\|v):=u\ln\frac uv+(1-u)\ln\frac{1-u}{1-v}\ ,\  0<u,v<1.
\]
The reverse ratio yields the MLSI lower bounds, while the convexity ratio yields the
entropy contraction upper bounds.  In continuous time,
\(\mathsf B_{e^{-t}}(C(E))\leq e^{-2\mathsf A(C(E))t}\) holds for every
\(t\ge0\), so the single channel contraction is stronger than the semigroup
exponential decay.  Conversely, under the parametrization \(p=e^{-t}\), the
infinitesimal decay rate of the second bound at \(t=0\) yields the MLSI lower bound obtained from the first estimate.
For a fixed-state depolarizing channel with
\(E(\rho)=\tr(\rho)\sigma\), these estimates are closely related to the sharp
depolarizing analysis of M\"uller-Hermes, Stilck Fran\c{c}a, and
Wolf~\cite{MullerHermesFrancaWolf2016}.  For general conditional
expectations, our estimates improve the curvature-based MLSI lower bound
\(\frac12+\frac{1}{C(E)+1}\) from~\cite{MunchWirthZhang2024} to
\(\frac{1}{2}+O(\frac{1}{\ln C(E)})\), which is asymptotically tight as shown above.

An advantage of our results is that they also yields complete entropy contraction estimate.  For classical
dynamics, the entropy contraction constants always tensorize
~\cite{DiaconisSaloffCoste1996,Raginsky2016}:
\[
 \alpha(P_t\otimes Q_t)
 =\min\{\alpha(P_t), \alpha(Q_t)\},
 \qquad
 \eta(\Phi_p\otimes\Psi_p)
 =\max\{\eta(\Phi_p), \eta(\Psi_p)\}.
\]
In the quantum setting, however, ordinary entropy contraction
rates need not be stable under tensoring with an ancilla.  Indeed,
ordinary quantum MLSI constants can decrease after adjoining a passive
ancilla, a phenomenon already observed for qubit depolarization
\cite[Proposition~4.21]{BrannanGaoJunge2022} (  See also
\cite[Problem~6.4]{GaoJungeLaRacuenteLi2025}) 
\[
 \frac12
 \le \alpha( P_t^{E_2}\otimes \id_{M_2})<0.965
 <1=\alpha(P_t^{E_2}).
\]
This phenomenon motivates the use of complete constants,
\begin{equation}
\alpha_{\mathrm c}(P_t^E)
\!:=\inf_{k\ge1}\alpha(P_t^E\otimes\id_{M_k}),
\qquad
\eta_{\mathrm c}(\Phi_p^E)
\!:=\sup_{k\ge1}\eta(\Phi_p^E\otimes\id_{M_k}).
\label{eq:gd-complete-constants}.
\end{equation}
These constant control
entropy contraction uniformly with arbitrary finite-dimensional
amplifications \cite{gao2020fisher,GaoJungeLaRacuenteLi2025}.

The tight estimate of the complete constants now follows immediately from 
applying Theorem \ref{cor:intro-entropy-contraction-bounds} to the complete index
\[
 C_{\mathrm{cb}}(E):=\sup\nolimits_{k\ge1}
 C(E\otimes\id_{M_k})\,,
\]
\begin{corollary}For a faithful conditional expectation $E$, 
\begin{align}
 \alpha(P_t^E)
  &=\frac12+\frac{1}{2\ln C_{cb}(E)}
    +O\!\left(\frac{\ln\ln C_{cb}(E)}{(\ln C_{cb}(E))^2}\right),\\
  \eta(\Phi_p^E)
  &=p+\frac{p\ln p}{\ln C_{cb}(E)}
    +O\!\left(\frac{\ln\ln C_{cb}(E)}{(\ln C_{cb}(E))^2}\right).
\end{align}\label{cor:cb}
\end{corollary}
Corollary~\ref{cor:cb} immediately gives the
system size independent tensorization property
\[
 \alpha_{\mathrm c}((P_t^E)^{\otimes n})
 =\alpha_{\mathrm c}(P_t^E),
 \qquad
 \eta_{\mathrm c}((\Phi_p^E)^{\otimes n})
 =\eta_{\mathrm c}(\Phi_p^E),
\]
Such a tensor stable property for the ordinary constant $\alpha$ and $\eta$ was only obtained for qubit depolarizing \cite{MullerHermesFrancaWolf2016} and is also studied under certain conditions in the manuscript in preparation \cite{HircheGeorgeNuradhaWilde2026}. 
The complete constants always tensorize stable under arbitrary
ancillary systems and suitable for quantum information tasks involving side
information.  Further examples include dephasing, where
\(D(\rho\|E(\rho))\) measures coherence
~\cite{BaumgratzCramerPlenio2014,StreltsovAdessoPlenio2017}, and symmetry
twirling, where the relative entropy measures asymmetry
~\cite{BartlettRudolphSpekkens2007}.

The rest of the paper is organized as follows.
Section~\ref{sec:ratio} provides the reverse and convexity ratios bounds and discusses their optimality.  Section~\ref{sec:gd}
is devoted to the entropy contraction of
generalized depolarization and tensorization property.  Examples are discussed in
Section~\ref{sec:examples}.  \\

\noindent {\bf Acknowledgement.} 
LG and LZ are partially supported by the National Natural Science Foundation
of China (grant no.~12401163) and the Department of Science and Technology
of Hubei Province (project nos.~2025EHA041 and 2025AFA044).
The authors acknowledge the use of AI tools during the exploratory stage of
this project.  All mathematical arguments and proofs in the final manuscript
were checked and written by the authors.

\section{Relative entropy Ratios}
\label{sec:ratio}
\subsection{Preliminaries on relative entropy}
\label{subsec:f-divergences}

For two quantum states with density operators \(\rho\) and \(\sigma\), the Umegaki relative entropy of $\rho$ with respect to $\sigma$ is
\[
 D(\rho\|\sigma):=
  \tr\!\bigl(\rho(\ln\rho-\ln\sigma)\bigr),
\]
provided $\operatorname{supp}(\rho)\subseteq\operatorname{supp}(\sigma)$, and $+\infty$ otherwise. 
The Umegaki relative entropy has wide applications in quantum information and also generalization as various quantum $f$-divergence \cite{HircheTomamichel2024}. \\

{\bf \noindent Hockey-Stick quantum $f$-divergence.} 
For our purpose, we recall the recently introduced Hockey-Stick quantum $f$-divergence and its basic properties \cite{BeigiHircheTomamichel2025}.
We will use the Riemann--Stieltjes representation by
Liu, Hirche, and Cheng
\cite[Propositions~6.1--6.2 and Lemma~6.3]{LiuHircheCheng2025}.  Throughout,
two self-adjoint operator \(X,Y\), we denote
\[
 \{X\le Y\}:=\mathds{1}_{[0,\infty)}(Y-X)
\]
by the spectral projection of \(Y-X\) onto its positive part. The symbol
\(\mathds{1}\) denotes the identity operator.
For faithful states $\rho$ and $\sigma$, define two distribution functions on $[0,\infty)$
\begin{align}
 P_{\rho,\sigma}(t)
 :=\tr\!\bigl(\rho\{\rho\le t\sigma\}\bigr), \qquad
 Q_{\rho,\sigma}(t)
:=\tr\!\bigl(\sigma\{\rho\le t\sigma\}\bigr), \qquad t\ge 0
 \label{eq:RS-distributions}
 \end{align}
 which satisfies
\begin{equation}
 dP_{\rho,\sigma}(t)
 =t\,dQ_{\rho,\sigma}(t).
\end{equation}
Then,  for a convex function \(f:[0,\infty)\to\mathbb R\) satisfying \(f(1)=0\), Hockey-Stick quantum $f$-divergence can be represented as 
\begin{equation}
 D_f(\rho\|\sigma)
 =\int_0^\infty f(t)\,
   dQ_{\rho,\sigma}(t).
 \label{eq:f-divergence-RS}
\end{equation}
This, via integration by part, is another quantum analog of classical $f$-divergence for two probability measures \(\mathsf P\ll\mathsf Q\), defined as
\[
 D_f(\mathsf P\|\mathsf Q)
 :=\int f\!\left(\frac{d\mathsf P}{d\mathsf Q}\right)\,d\mathsf Q.
\]
For \(f\) twice continuously differentiable on \((0,\infty)\),
\eqref{eq:f-divergence-RS} agrees with the original Hockey-Stick integral formulation
of~\cite{HircheTomamichel2024}; see
\cite[Theorem~3.1 and Proposition~6.1]{LiuHircheCheng2025}.

The Umegaki relative entropy corresponds to the function $f(t)=t\ln t$. Note that adding a multiple of \(t-1\) to the function does not change \(D_f\) on states, as
\[ \int_0^\infty (t-1)\,
   dQ_{\rho,\sigma}(t)=\int_0^\infty \,
   dP_{\rho,\sigma}(t)-\int_0^\infty \,
   dQ_{\rho,\sigma}(t)=\tr(\rho)-\tr(\sigma)=0. \]
Therefore, we have
\begin{equation}
 D(\rho\|\sigma)= D_\phi(\rho\|\sigma)
 =\int_0^\infty \phi(t)\,
   dQ_{\rho,\sigma}(t).
 \label{eq:relative-entropy-RS}
\end{equation}
via the normalized function 
\begin{equation}
 \phi(t):=t\ln t-t+1,
 \qquad t>0,
 \label{eq:relative-generator}
\end{equation}
which is $f$ subtracting the first order term $f'(1)(t-1)$ at $t=1$.

One advantage of the Riemann--Stieltjes representation is that interchanging
the two states is described by the transposed function.
\begin{equation}
 D_f(\rho\|\sigma)=D_{f^\diamond}(\sigma\|\rho),\quad \text{where}\quad f^\diamond(t):=t f(t^{-1}).
 \label{eq:f-exchange}
\end{equation}
This is the \(\star\)-conjugation described in
\cite[Proposition~2.6(3)]{HircheTomamichel2024}; see also
\cite[proof of Corollary~3.8]{BeigiHircheTomamichel2025}.  In particular,
\(\phi^\diamond(t)=-\ln t-1+t\), and hence
\begin{equation}
 D(\sigma\|\rho)=D_{\phi^\diamond}(\rho\|\sigma).
 \label{eq:reverse-generator}
\end{equation}
 
We also record the two convexity functions. Let \(\rho,\sigma\) be faithful states, let \(0<p<1\), and define the convex combination
\[
 \rho_p=p\rho+(1-p)\sigma.
\]
The change-of-reference rule for the hockey-stick \(f\)-divergence gives (see \cite[Proposition~4.6]{HircheTomamichel2024} and
\cite[Eq.~(5.22)]{BeigiHircheTomamichel2025})
\[
 D_f(p\rho+(1-p)\sigma\|\sigma)=D_{f\circ a_p}(\rho\|\sigma);
\]
where 
\[a_p(t)=1+p(t-1)\]
In particular, for relative entropy, taking \(f=\phi\) yield
\begin{align}
D(p\rho+(1-p)\sigma\|\sigma)
 =D_{f_p}(\rho\|\sigma), \label{eq:output-generator}\\ 
D(\rho \|p\rho+(1-p)\sigma)
 =D_{g_p}(\rho\|\sigma), \label{eq:input-generator}
\end{align}
where 
\[f_p(t):=\phi(a_p(t)),  \]
and 
\[ g_p(t)
 :=a_p(t)\phi\!\left(\frac{t}{a_p(t)}\right)
 =t\ln\frac{t}{a_p(t)}-(1-p)(t-1). \]
Here, the second change-of-reference rule gives
\[
 D_f(\rho\|\rho_p)
 =D_{a_p(\,\cdot\,)f(\,\cdot\,/a_p(\,\cdot\,))}(\rho\|\sigma);
\]
see \cite[Proposition~2.9]{HircheTomamichel2024} and
\cite[Eq.~(5.3)]{BeigiHircheTomamichel2025}.  \\

{\bf \noindent Bogoliubov--Kubo--Mori metric.} We will also use the integral representation of relative entropy via Bogoliubov--Kubo--Mori (BKM) quantum information metric. For a positive definite operator \(\rho\), the Bogoliubov--Kubo--Mori (BKM)
metric for a operator $X$ is defined by
\begin{equation}
 \|X\|_{\rho^{-1}}^2
 :=\int_0^\infty
 \tr\!\bigl(X^*(\rho+r\mathds 1)^{-1}
 X(\rho+r\mathds 1)^{-1}\bigr)\,dr .
 \label{eq:bkm-metric}
\end{equation}

\begin{lemma}[BKM order comparison \cite{GaoRouze2022}]\label{lem:compare}
Let \(\rho,\sigma\) be two positive operators  with \(\rho\le c\,\sigma\) for some \(c>0\). Then for any operator \(X\),
\[
\|X\|_{\sigma^{-1}}^2
\le c\|X\|_{\rho^{-1}}^2 .
\]
\end{lemma}
\begin{proof}
Using the cyclicity of the trace and the operator anti-monotonicity of \(t\mapsto t^{-1}\),
\begin{align*}
\int_{0}^\infty \tr\bigl(X^*(\rho+r\mathds 1)^{-1}X(\rho+r\mathds 1)^{-1}\bigr)\,dr  \ge &  \int_{0}^\infty \tr\bigl(X^*(c\sigma+r\mathds 1)^{-1}X(c\sigma+r\mathds 1)^{-1}\bigr)\,dr \\
= &\frac1c\int_{0}^\infty \tr\bigl(X^*(\sigma+r\mathds 1)^{-1}X(\sigma+r\mathds 1)^{-1}\bigr)\,dr .
\end{align*}
The last equality follows from the change of variables \(r\mapsto r/c\).
\end{proof}

We also recall Chebyshev's integral inequality.
\begin{lemma}[Chebyshev's integral inequality]
\label{lem:chebyshev}
Let \(\mu\) be a finite nonzero positive measure on \([0,1]\), and let \(f,g:[0,1]\to \mathbb{R}\) be integrable functions with the same monotonicity. Then
\[
 \mu([0,1])\int_0^1fg\,d\mu
 \ge
 \left(\int_0^1f\,d\mu\right)
 \left(\int_0^1g\,d\mu\right).
\]
The inequality is reversed when \(f\) and \(g\) have opposite monotonicity.
\end{lemma}

\begin{proof}
Writing \(m=\mu([0,1])\), one has
\[
 m\int fg\,d\mu-\int f\,d\mu\int g\,d\mu
 =
 \frac12\iint
 (f(s)-f(t))(g(s)-g(t))\,d\mu(s)d\mu(t).
\]
The sign follows immediately from the assumed monotonicities.
\end{proof}

We first assume that \(\rho\) and \(\sigma\) have common support. By restricting them
to that support, we can assume $\rho$ and $\sigma$ are positive definite.  Recall that we write \(\rho_p=p\rho+(1-p)\sigma\) for $p\in[0,1]$.  Differentiating relative entropy
along this  linear segment and integrating twice gives (see \cite{GaoRouze2022} for the detail)
\begin{equation}
 \begin{aligned}
 D(\rho\|\sigma)
 &=\int_0^1(1-s)\|\rho-\sigma\|_{\rho_s^{-1}}^2\,ds,\\
 D(\sigma\|\rho)
 &=\int_0^1s\|\rho-\sigma\|_{\rho_s^{-1}}^2\,ds,\\
 D(p\rho+(1-p)\sigma\|\sigma)
 &=\int_0^p(p-s)\|\rho-\sigma\|_{\rho_s^{-1}}^2\,ds,\\
 D(\rho\|p\rho+(1-p)\sigma)
 &=\int_p^1(1-s)\|\rho-\sigma\|_{\rho_s^{-1}}^2\,ds.
 \end{aligned}
 \label{eq:repD}
\end{equation}

\subsection{Reverse and convexity ratios}
\label{subsec:key-comparison}
For two density operators \(\rho\) and \(\sigma\), we define
\[
 C:=\inf\{c>0:\rho\le c\sigma\}<\infty.
\]
Clearly, \(C\ge1\), with equality if and only if \(\rho=\sigma\).  This is closely related to the max-relative
entropy introduced by Datta ~\cite{datta2009min},
\[
 D_{\max}(\rho\|\sigma)
 :=\ln\inf\{c>0:\rho\le c\sigma\}
 =\ln C.
\]
The next theorem gives universal bounds for the reverse ratio and convexity ratios of relative entropy depending
only on this constant.

\begin{proposition}[Reverse and convexity ratios]
\label{thm:relative-entropy-comparisons}
Let \(\rho\ne\sigma\) be density operators with minimal order constant
\[
 C:=\inf\{c>0:\rho\le c\sigma\}>1.
\]
 Then
\begin{enumerate}
\item[i)] 
\begin{equation}
 \frac{D(\sigma\|\rho)}{D(\rho\|\sigma)}
 \ge
 \frac{C-1-\ln C}{C\ln C-C+1}.
 \label{eq:reverse-comparison}
\end{equation}
\item[ii)] for \(0<p<1\),
\begin{align}
 \phi(1-p)
 \le
 \frac{D(p\rho+(1-p)\sigma\|\sigma)}{D(\rho\|\sigma)}
 \le
 \frac{\phi(a_p(C))}{\phi(C)},
 \label{eq:output-convexity-comparison}\\
 \frac{a_p(C)\phi(C/a_p(C))}{\phi(C)}
 \le
 \frac{D(\rho\|p\rho+(1-p)\sigma)}{D(\rho\|\sigma)}
 \le 1-p .
 \label{eq:input-convexity-comparison}
\end{align}
where
\[
 \phi(t)=t\ln t-t+1,
 \qquad
 a_p(t)=1+p(t-1).
\]
\end{enumerate}
\end{proposition}

We present two independent proofs for the ratio bounds above. The first one follows from the comparison of Hockey-Stick $f$-divergence from \cite[Theorem~4.2]{BeigiHircheTomamichel2025}. We also provide a second proof using the BKM integral representation~\eqref{eq:repD} following the idea of \cite{GaoRouze2022}. For both approach, we can assume the density operators $\rho$ and $\sigma$ has common support, hence can be viewed as positive definite on their support. The general case follows by faithful perturbation
\(\rho_\varepsilon=(1-\varepsilon)\rho+\varepsilon\sigma\) and then taking
\(\varepsilon\downarrow0\), where the corresponding order constant is
\(C_\varepsilon=(1-\varepsilon)C+\varepsilon\) converging to \(C\).

\begin{proof}[Proof via HS $f$-divergence]
This theorem follows the idea of
\cite[Theorem~4.2]{BeigiHircheTomamichel2025}.
By the discussion above, it suffices to work on $\rho$ and $\sigma$ with common
support. Set
\[
 r(t):=\frac{\phi^\diamond(t)}{\phi(t)},
 \qquad
 R_p(t):=\frac{f_p(t)}{\phi(t)},
 \qquad
 S_p(t):=\frac{g_p(t)}{\phi(t)},
\]
for \(t\ne1\), and
 \(r(1):=1\), \(R_p(1):=p^2\), and \(S_p(1):=(1-p)^2\) by continuity. 
The Riemann--Stieltjes integral representation reduce all
three entropy ratios to extrema of the quotients of the corresponding scalar functions:
\begin{align}
 \frac{D(\sigma\|\rho)}{D(\rho\|\sigma)}
 =&\frac{\int_{(0,C]}\phi^\diamond(t)\,dQ_{\rho,\sigma}(t)}{\int_{(0,C]}\phi(t)\,dQ_{\rho,\sigma}(t)}\geq \inf_{t\in(0,C]}r(t)
 \label{eq:reverse-generator-ratio}\\
 \sup_{t\in(0,C]}R_p(t)\geq \frac{D(\rho_p\|\sigma)}{D(\rho\|\sigma)}
 =&\frac{\int_{(0,C]}f_p(t)\,dQ_{\rho,\sigma}(t)}{\int_{(0,C]}\phi(t)\,dQ_{\rho,\sigma}(t)}\geq \inf_{t\in(0,C]}R_p(t),
 \label{eq:output-generator-ratio}\\
\sup_{t\in(0,C]}S_p(t)\geq \frac{D(\rho\|\rho_p)}{D(\rho\|\sigma)}
 =&\frac{\int_{(0,C]}g_p(t)\,dQ_{\rho,\sigma}(t)}{\int_{(0,C]}\phi(t)\,dQ_{\rho,\sigma}(t)}\geq \inf_{t\in(0,C]}S_p(t),
 \label{eq:input-generator-ratio}
\end{align}
Note that here $dQ_{\rho,\sigma}$ is supported on $[0,C]$ because of the order condition. 
For the reverse ratio, the scalar quotient is
\[
 r(t)=\frac{-\ln t-1+t}{t\ln t-t+1},
\]
For \(t\ne1\), direct
differentiation gives
\[
 r'(t)
 =\frac{t(\ln t)^2-(t-1)^2}{t\phi(t)^2}\le0,
\]
where the inequality follows from
\(\sqrt t\,|\ln t|\le|t-1|\).  Thus \(r\) is non-increasing on
\((0,\infty)\), and
\[
 \inf_{t\in(0,C]}r(t)=r(C)=\frac{C-1-\ln C}{C\ln C-C+1},
\]
Equation~\eqref{eq:reverse-generator-ratio} therefore proves
\eqref{eq:reverse-comparison}.

For the first convexity ratio, we show that 
\[R_p(t)=\frac{f_p(t)}{\phi(t)} \]
is monotone increasing over $[0,C]$, which can be argued on $[0,1)$ and $(1,C]$ separately.  Denote 
\[ H_p(t)=\frac{f_p''(t)}{\phi''(t)}
 =\frac{p^2t}{a_p(t)}\]
 as the quotient of second order derivatives of numerator and denominator respectively. Note that
\[ \left(\frac{p^2t}{a_p(t)}\right)'
 =\frac{p^2(1-p)}{a_p(t)^2}>0,
\]
 $H_p$ is strictly increasing on \([0,C]\). 
Note that \[ f_p(1)=\phi(1)= f_p'(1)=\phi'(1)=0.\]
Applying Monotone Form of L'H\^opital's Rule twice implies \(R_p\) is increasing on each of
the full intervals \((0,1)\) and \((1,C]\), which yields the convexity bound \eqref{eq:output-generator-ratio} by put $t=0$ and $t=C$.

The argument for the second convexity ratio is similar. Just note
\[ g_p(1)=\phi(1)= g_p'(1)=\phi'(1)=0, \]
and the quotient of second derivatives is monotone decreasing
\[ \frac{g_p''(t)}{\phi''(t)}
 =\frac{(1-p)^2}{a_p(t)^2},
 \qquad
 \left(\frac{(1-p)^2}{a_p(t)^2}\right)'
 =-\frac{2p(1-p)^2}{a_p(t)^3}<0.\]
Therefore, 
 the bounds in \eqref{eq:input-convexity-comparison} follows from 
putting $t=C$ and $t=0$ into $S_p$ for the minimum and maximum on $[0,C]$ respectively. 
\end{proof}
We now present the alternative proof by BKM metric
\begin{proof}[Proof via BKM metric]
\label{subsec:bkm-approach}Recall that $\rho_s=s\rho+(1-s)\sigma$. 
For \(s\in[0,1]\), define
\[
 h(s)
 =\bigl(1+(C-1)s\bigr)\|\rho-\sigma\|_{\rho_s^{-1}}^2.
\]
Note that $ h$ is monotone non-decreasing, because of Lemma~\ref{lem:compare} and  \(0\le s<t\le1\),
\[
 \bigl(1+(C-1)t\bigr)\rho_{s}
 -\bigl(1+(C-1)s\bigr)\rho_{t}
 =(t-s)(C\sigma-\rho)\ge0.
\]
Define
\[
 d\mu(s)=\frac{ds}{1+(C-1)s}.
\]
The first two identities in \eqref{eq:repD} become
\[
 D(\rho\|\sigma)=\int_0^1(1-s)h(s)\,d\mu(s),
 \qquad
 D(\sigma\|\rho)=\int_0^1s h(s)\,d\mu(s).
\]
Applying Chebyshev integral inequality Lemma~\ref{lem:chebyshev} to the pairs \((h,s)\) and \((h,1-s)\)
gives
\[
 \frac{D(\sigma\|\rho)}{D(\rho\|\sigma)}
 \ge
 \frac{\int_0^1s\,d\mu(s)}
      {\int_0^1(1-s)\,d\mu(s)}
 =
 \frac{C-1-\ln C}{C\ln C-C+1}.
\]
For convexity ratios,
set
\[
 r_p(s)=\frac{(p-s)_+}{1-s}, \qquad 0\le s<1
 \qquad \text{ and } r_p(1)=0.
\]
For the upper bound, define the measure
\[
 d\nu(s)=\frac{1-s}{1+(C-1)s}\,ds.
\]
As shown above, \(h\) is non-decreasing, while \(r_p\) is
non-increasing on \([0,1]\).  Hence by 
Lemma~\ref{lem:chebyshev} again and \eqref{eq:repD},
\[
 \frac{D(\rho_p\|\sigma)}{D(\rho\|\sigma)}
 =\frac{\int_0^1h(s)r_p(s)\,d\nu(s)}
        {\int_0^1h(s)\,d\nu(s)}
 \le
 \frac{\int_0^1r_p(s)\,d\nu(s)}{\nu([0,1])}.
\]
The integrals are
\[
 \nu([0,1])=\frac{\phi(C)}{(C-1)^2},
 \qquad
 \int_0^1r_p(s)\,d\nu(s)
 =\frac{\phi(a_p(C))}{(C-1)^2},
\]
which gives the upper bound in \eqref{eq:output-convexity-comparison}.

For the lower bound, set
\[  F(s)=(1-s)\|\rho-\sigma\|_{\rho_s^{-1}}^2 .\]
For \(0\le s<t<1\),
\[ 
 (1-s)\rho_t-(1-t)\rho_s
 =(t-s)\rho\ge0.
\]
By Lemma~\ref{lem:compare}, this implies
\[
 F(s)=(1-s)\|\rho-\sigma\|_{\rho_s^{-1}}^2\ge (1-t)\|\rho-\sigma\|_{\rho_t^{-1}}^2=F(t).\]  
Therefore, \(F\) is non-increasing. Since \(r_p\) is also non-increasing,
by Chebyshev integral inequality (Lemma~\ref{lem:chebyshev}) applied to Lebesgue measure, we have
\[
 \frac{D(\rho_p\|\sigma)}{D(\rho\|\sigma)}
 =\frac{\int_0^1r_p(s)F(s)\,ds}{\int_0^1F(s)\,ds}
 \ge\int_0^1r_p(s)\,ds
 =\phi(1-p),
\]
which completes \eqref{eq:output-convexity-comparison}.

For the second convexity ratios, 
 we first have 
\[
 \frac{D(\rho\|\rho_p)}{D(\rho\|\sigma)}
 =\frac{\int_0^1\mathbf{1}_{[p,1]}(s)F(s)\,ds}
        {\int_0^1F(s)\,ds}\le 1-p.
\]
Here we use, again,  the Chebyshev integral inequality for the non-decreasing indicator $\mathbf{1}_{[p,1]}$ and \(F\) being non-increasing.  On the other hand,
\(h\) and \(\mathbf{1}_{[p,1]}\) are both non-decreasing, and hence
\[
 \frac{D(\rho\|\rho_p)}{D(\rho\|\sigma)}
 =\frac{\int_0^1\mathbf{1}_{[p,1]}(s)h(s)\,d\nu(s)}
        {\int_0^1h(s)\,d\nu(s)}
 \ge
 \frac{\int_p^1d\nu(s)}{\nu([0,1])}.
\]
The calculation gives
\[
 \int_p^1d\nu(s)
 =\frac{a_p(C)\phi(C/a_p(C))}{(C-1)^2},
\]
which gives \eqref{eq:input-convexity-comparison} and completes the 
proof via BKM metric.
\end{proof}

\begin{remark}[Optimality and sharpness]{\rm 
In the classical setting, the optimal
comparison constants of $f$-divergence are obtained via  functional domination in
\cite[Theorem~6 and Remarks~11--12]{SasonVerdu2016}.  In particular, the
reverse comparison in Theorem~\ref{thm:relative-entropy-comparisons} is the
one-sided specialization of
\cite[Theorem~7]{SasonVerdu2016}. Given the assumption \(\rho\le C\sigma\), these bounds are known to be optimal already in classical cases, hence so are them for quantum states. 
For the Hockey-Stick quantum
\(f\)-divergence, the general quantum functional-domination theorem was obtained in
\cite[Theorem~4.2]{BeigiHircheTomamichel2025} under the common domain assumption. 
The first Hockey-Stick $f$-quantum divergence argument basically combined this result with extrema analysis for the quotients of the pairs
$
 (\phi^\diamond,\phi),
 (f_p,\phi)$, and $
 (g_p,\phi)$.}
\end{remark}

\section{Entropy contraction of generalized depolarization}
\label{sec:gd}
This section studies entropy contraction generated by a faithful conditional
expectation on a finite-dimensional quantum system. We use the map 
\[
E^*:\mathcal B(\mathcal H)\longrightarrow\mathcal N
\subseteq\mathcal B(\mathcal H)
\]
with star notation to denote the unital completely positive map in the Heisenberg picture and $E$ for
its Schr\"odinger picture channel on states.  A star therefore indicates a map
on observables, while the corresponding unstarred map acts on states.  We work
below mainly in the Schr\"odinger picture.  
The continuous and one-step depolarizing dynamics are
\[
P_t^E=e^{-t}\id+(1-e^{-t})E,
\qquad
\Phi_p^E=p\,\id+(1-p)E\,,
\]
The main purpose is to apply the three
ratio comparisons from the previous section to these conditional-expectation
dynamics.  The operator order relation
\[
  \rho\le C(E)E(\rho)
\]
provides precisely the order condition needed there, with the second state chosen
as $E(\rho)$.  The constant $C(E)$ is the least number such that
$\rho\le C(E)E(\rho)$ for every state $\rho$; its index-theoretic origin is
recalled below.

\subsection{Conditional expectations and index bounds}
\label{subsec:gd-index}

Let $\mathcal N\subseteq\mathcal B(\mathcal H)$ be a unital
$*$-subalgebra.  In the Heisenberg picture, a conditional expectation onto
$\mathcal N$ is a unital completely positive idempotent map onto $\mathcal N$
\[
 E^*:\mathcal B(\mathcal H)\longrightarrow\mathcal B(\mathcal H),
 \qquad (E^*)^2=E^*,
 \qquad \operatorname{Ran}E^*=\mathcal N.
\]
By \cite{Takesaki72}, such map satisfies the bimodule
property
\begin{equation}
 E^*(axb)=aE^*(x)b,
 \qquad a,b\in\mathcal N,\quad x\in\mathcal B(\mathcal H).
\label{eq:gd-bimodule}
\end{equation}
We call $E^*$ faithful if
\[
 E^*(x^*x)=0\quad\Longrightarrow\quad x=0.
\]
The Schr\"odinger picture adjoint $E$ is defined via duality
\begin{equation}
  \tr(E(\rho)x)=\tr(\rho E^*(x)),
  \qquad
  \rho\in\mathcal D(\mathcal H),\quad x\in\mathcal B(\mathcal H)\,.
\label{eq:gd-preadjoint}
\end{equation}
Hence $E$ is a completely positive trace-preserving channel and is
idempotent $E^2=E$. Every state $\sigma\in \operatorname{Ran}E$ is a invariant state  $\sigma=E(\sigma)$. 
Moreover, if it admits a faithful
invariant state $\sigma$ so that $E(\sigma)=\sigma$ and $\sigma>0$, then
$E^*$ is automatically faithful. Indeed,
\[
 \tr\!\left(\sigma E^*(x^*x)\right)
 =\tr\!\left(E(\sigma)x^*x\right)
 =\tr(\sigma x^*x),
\]
and the right-hand side vanishes only when $x=0$.  

The conditional-expectation satisfies the chain rule for relative entropy
\begin{equation}
  D(\rho\|\omega)
  =
  D(\rho\|E(\rho))
  +
  D(E(\rho)\|\omega),
  \qquad \omega\in\mathcal D(E)\,,
  \label{eq:gd-pythagorean}
\end{equation}
and hence
\[
  D(\rho\|E(\rho))
  =
  \inf_{\omega=E(\omega)}D(\rho\|\omega)\,,
\]
with the minimum attained at $E(\rho)$
\cite{junge2019stability}.

The crucial quantity we care about here for a conditional expectation is the Pimsner and Popa index \cite{PimsnerPopa1986}.
\begin{definition}[Pimsner--Popa index]
For a faithful conditional expectation channel $E$, define
\begin{align}
C(E)
&:=
\inf\{C>0:\rho\le CE(\rho)
       \text{ for every state }\rho\},
\label{eq:gd-index}\\
C_{\mathrm{cb}}(E)
&:=
\sup_{k\ge1}C(E\otimes\id_{M_k})\,.
\label{eq:gd-cb-index}
\end{align}
\end{definition}
The constant $C(E)$ was first introduced by Pimsner and Popa~\cite{PimsnerPopa1986} as finite-dimensional variants of the subfactor index initiated by Jones~\cite{Jones1983}. Its
completely bounded version $C_{\mathrm{cb}}(E)$ was studied
in~\cite{gaoindex} and subsequently used for complete entropic inequalities
in~\cite{GaoRouze2022}.

In finite dimensions, a faithful conditional expectation $E^*:\mathcal B(\mathcal H)\longrightarrow\mathcal N$, up to unitary equivalence,
admits the following decompositions \cite{GaoRouze2022}
\[
\mathcal H=\bigoplus_{i=1}^n\mathcal H_i\otimes\mathcal K_i,
\qquad
\mathcal N=\bigoplus_{i=1}^n
\mathcal B(\mathcal H_i)\otimes\mathbb C\mathds 1_{\mathcal K_i},
\]
and 
\begin{align}
E^*(x)
&=\bigoplus_{i=1}^n
\tr_{\mathcal K_i}\!\left[
 (\mathds 1_{\mathcal H_i}\otimes\tau_i)P_ixP_i
\right]\otimes\mathds 1_{\mathcal K_i},
\label{eq:gd-E-star-block}\\
E(\rho)&=\bigoplus_{i=1}^n
\tr_{\mathcal K_i}(P_i\rho P_i)\otimes\tau_i\,,
\label{eq:gd-Estar-block}
\end{align}
where $\tau_i\in\mathcal D(\mathcal K_i)$ are some faithful states.
Moreover, a state $\sigma$ is preserved by $E$ if and only if
\begin{equation}
 \sigma=\bigoplus_{i=1}^n p_i\,\sigma_i\otimes\tau_i,
 \qquad
 p_i\ge0,\quad \sum_i p_i=1,\quad
 \sigma_i\in\mathcal D(\mathcal H_i).
\label{eq:gd-invariant-state-block}
\end{equation}
Such a state $\sigma$ is faithful precisely when $p_i>0$ and $\sigma_i>0$ for every $i$.

We record the finite dimension formulas for Pimsner-Popa index below. 

\begin{proposition}[Index formula] Let $E$ be a faithful conditional expectation with decomposition as above. Let $m_i=\dim\mathcal H_i$, $n_i=\dim\mathcal K_i$, and let
$s_{i,1}\le\cdots\le s_{i,n_i}$ be the eigenvalues of $\tau_i$.  Then,
\label{prop:gd-index-structure}
\begin{equation}
C(E)=\sum_{i=1}^n\sum_{j=1}^{\min(m_i,n_i)}\frac1{s_{i,j}},
\qquad
C_{\mathrm{cb}}(E)=\sum_{i=1}^n\sum_{j=1}^{n_i}\frac1{s_{i,j}}\,.
\label{eq:gd-index-formulas}
\end{equation}
In particular, $C(E)=C_{\mathrm{cb}}(E)$ whenever $m_i\ge n_i$ in every
block.
\end{proposition}
\begin{proof}
 By convexity it suffices
to consider pure state $\rho=|\psi\rangle\langle\psi|$.  Writing
$\psi_i=P_i\psi$ and
$\sigma_i=\tr_{\mathcal K_i}|\psi_i\rangle\langle\psi_i|$, set
$A_\psi=E(|\psi\rangle\langle\psi|)
=\bigoplus_i\sigma_i\otimes\tau_i$.  Rank one domination yields
\begin{align}
 |\psi\rangle\langle\psi|\le cA_\psi
 \quad\Longleftrightarrow\quad
 c\ge\sum_{i:\psi_i\ne0}
 \langle\psi_i|(\sigma_i^{-1}\otimes\tau_i^{-1})|\psi_i\rangle, \label{eq:indexbound} 
\end{align}
where inverses are restricted to the relevant supports.

To make the direct-sum contribution explicit, write
$\psi_i=\sqrt{p_i}\,\varphi_i$, where
$p_i=\|\psi_i\|^2$, $\sum_i p_i=1$, and $\|\varphi_i\|=1$ whenever
$p_i>0$.  If
$\eta_i=\tr_{\mathcal K_i}|\varphi_i\rangle\langle\varphi_i|$, then
$\sigma_i=p_i\eta_i$, and hence
\[
 \langle\psi_i|(\sigma_i^{-1}\otimes\tau_i^{-1})|\psi_i\rangle
 =
 \langle\varphi_i|(\eta_i^{-1}\otimes\tau_i^{-1})|\varphi_i\rangle.
\]
The weight $p_i$ therefore cancels in every nonzero block, hence does not affect the index lower bound in \eqref{eq:indexbound}.  We may choose all
$p_i>0$ and optimize the normalized vectors $\varphi_i$ independently, so
the supremum over the direct sum is the sum of the block suprema.

For one block, let
$\varphi_i=\sum_{j=1}^r\sqrt{\lambda_j}\,u_j\otimes v_j$ be a Schmidt
decomposition.  Its contribution is
\[
 \sum_{j=1}^r\langle v_j,\tau_i^{-1}v_j\rangle,
 \qquad r\le\min(m_i,n_i).
\]
By the variational principle, its maximum is the sum of the
$\min(m_i,n_i)$ largest eigenvalues of $\tau_i^{-1}$, namely
$\sum_{j=1}^{\min(m_i,n_i)}s_{i,j}^{-1}$.  Summing these block optima gives
the formula for $C(E)$.

For $E\otimes\id_{M_k}$, the dimension $m_i$ is replaced by $km_i$.
For all sufficiently large $k$ such $km_i\ge n_i$ for every $i$, every eigenvalue of $\tau_i^{-1}$
contributes, and
\[
 C_{\mathrm{cb}}(E)=\sum_i\tr(\tau_i^{-1})
 =\sum_i\sum_{j=1}^{n_i}s_{i,j}^{-1}.\qedhere
\]
\end{proof}

\begin{remark}{\rm 
For the trace-preserving conditional expectation, one has
$\tau_i=\mathds 1_{\mathcal K_i}/n_i$, recovering the usual
Pimsner--Popa formula \cite[Theorem 6.1]{PimsnerPopa1986}.
}
\end{remark}
\subsection{Modified logarithmic Sobolev inequalities}
\label{subsec:gd-mlsi}

In finite dimensions, a quantum Markov semigroup (QMS) is a continuous
family of quantum channels $(P_t)_{t\ge0}$ satisfying
$P_0=\id$ and $P_{t+s}=P_tP_s$.  Its generator $\mathcal L$, called Lindbladian, is
defined by $\mathcal{L}=\lim_{t\to 0}\frac{1}{t}(P_t-\id), P_t=e^{t\mathcal L}$. Here we consider the generalized depolarizing semigroup
\[P_t^E=e^{-t}\id+(1-e^{-t})E, \mathcal L=E-\id,\]
associated to a faithful conditional expectation channel $E$.

For example, for a faithful state $\sigma$ and the expectation map $E_\sigma(\rho)=\tr(\rho)\sigma$,
\[P_t(\rho)=e^{-t}\rho+(1-e^{-t})\tr(\rho)\sigma,\]
is the $\sigma$-biased depolarizing semigroup. In general, when $\mathcal{N}\neq \mathbb{C}1$ is not-trivial, $P_t^E$ is in general non-primitive, i.e., every initial state $\rho$ converges to its equilibrium state $E(\rho)$ which are not necessarily unique. Here, we are in particular interested in exponential decay of relative entropy of $P_t^E(\rho)$ with respect to $E(\rho)$. Such quantitative estimate are studied through modified logarithmic Sobolev
inequalities (MLSI) (see e.g. \cite{Bardet17}).

\begin{definition}[MLSI constant]
The modified logarithmic Sobolev constant $\alpha(P_t^E)$ is the largest
$\alpha\ge0$ such that for every state $\rho$ and every $t\ge0$,
\begin{equation}
D(P_t^E(\rho)\|E(\rho))
\le e^{-2\alpha t}D(\rho\|E(\rho)).
\label{eq:gd-MLSI}
\end{equation}
\end{definition}

The term MLSI originally refers to the equivalent differential form of
\eqref{eq:gd-MLSI}.  Differentiating \eqref{eq:gd-MLSI} at $t=0$ gives
\[
 2\alpha D(\rho\|E(\rho))\leq\mathcal I_E(\rho).
\]
where 
\begin{align}
\mathcal I_E(\rho)
&:=-\left.\frac d{dt}\right|_{t=0}
D(P_t^E(\rho)\|E(\rho))\nonumber=\tr\!\left[(\rho-E(\rho))
(\ln\rho-\ln E(\rho))\right].
\label{eq:gd-entropy-production}
\end{align}
is the entropy production of the generator $E-\id$ at $\rho$, also called the
relative Fisher information. Note that 
\[
 \mathcal I_E(\rho)
 =D(\rho\|E(\rho))+D(E(\rho)\|\rho).
\]
Consequently, the optimal constant
\begin{equation}
\alpha(P_t^E)
=\frac{1}{2}\inf_{\rho\ne E(\rho)}
\frac{\mathcal I_E(\rho)}{D(\rho\|E(\rho))}=\frac12\bigl(1+\inf_{\rho\ne E(\rho)}
\frac{D(E(\rho)\|\rho)}{D(\rho\|E(\rho))}\bigr) \,,
\label{eq:gd-alpha-reverse}
\end{equation}
Hence, the infinitesimal MLSI constant is exactly a reverse relative-entropy
ratio.  Moreover, the index relation $\rho\le C(E)E(\rho)$ allows us to
apply the reverse comparison in
Theorem~\ref{thm:relative-entropy-comparisons} with $\sigma=E(\rho)$ and
$C=C(E)$.

\begin{corollary}[MLSI lower bound]
\label{cor:gd-MLSI}
For every faithful conditional-expectation channel $E$,
\begin{equation}
\alpha(P_t^E)\ge
\frac12\left(1+\frac{C(E)-\ln C(E)-1}{C(E)\ln C(E)-C(E)+1}\right)\,.
\label{eq:gd-MLSI-bound}
\end{equation}
\end{corollary}
\begin{remark}[Comparison to Curvature bound] {\rm 
For the semigroup $P_t^E$ generated by a conditional
expectation, M\"unch, Wirth, and Zhang ~\cite{MunchWirthZhang2024} obtained intertwining curvature bound, which by quantum Bakry-Emery theorem \cite{DattaRouze18,CarlenMaas18} implies the following entropy contraction rate
\[
 \alpha(P_t^E)\ge\frac12+\frac1{C(E)+1}.
\]
Our entropy-ratio estimate improves the correction of order \(C(E)^{-1}\)
to \((\ln C(E))^{-1}\), which as shown below is also tight. }
\end{remark}

\subsection{Convexity ratios and entropy contraction}
\label{subsec:gd-sdpi}
Let $E$ be a faithful conditional expectation channel. For $0<p<1$, consider the generalized depolarizing channel
\begin{equation}
\Phi_p^E(\rho)=p\rho+(1-p)E(\rho)\,.
\label{eq:gd-Phi-star}
\end{equation}
This channel is a single time map of the corresponding semigroup:
\begin{equation}
\Phi_p^E=P_{-\ln p}^E\,,
\label{eq:gd-continuous-discrete}
\end{equation}
where the discrete parameter $p$ is simply the retained weight of the initial
state at time $t=-\ln p$.
\begin{definition}[Entropy contraction constant]
The relative entropy contraction constant associated with the fixed-point
channel \(E\) is
\begin{equation}
\eta(\Phi_p^E):=
\sup_{\rho\ne E(\rho)}
\frac{D(\Phi_p^E(\rho)\|E(\rho))}
     {D(\rho\|E(\rho))}=\sup_{\rho\ne E(\rho)}
\frac{D(p\rho+(1-p)E(\rho)\|E(\rho))}
     {D(\rho\|E(\rho))}\,.
\label{eq:gd-eta}
\end{equation}
\end{definition}
Hence the finite time entropy contraction problem is a convexity ratio.  As in the MLSI
argument, the index relation allows us to apply the convexity comparison
in Theorem~\ref{thm:relative-entropy-comparisons}.
\begin{corollary}[Entropy contraction]
\label{cor:gd-entropy-contraction}
Let $C=C(E)$ and $a_p(C)=1+p(C-1)$.  Then
\[
 \eta(\Phi_p^E)
 \le \frac{\phi(a_p(C))}{\phi(C)}\, ,
\]
where \(\phi(x)=x\ln x-x+1\).
\end{corollary}

The convexity and MLSI entropy contraction estimates are compatible.
\begin{remark}[Compatibility at infinitesimal time]{\rm Write the index $C=C(E)$.
Recall that the MLSI bound gives the continuous time entropy contraction
\begin{equation}
\eta(P_t^E)\le e^{-2\alpha(C)t} , \qquad \alpha(C)=\frac12\left(1+\frac{C-\ln C-1}{C\ln C-C+1}\right)
\label{eq:gd-MLSI-to-entropy-contraction}
\end{equation} 
The convexity ratio estimate gives
\[
  \eta(\Phi_p^E)\le \frac{\phi(1+p(C-1))}{\phi(C)} \,.
\]
For $p=e^{-t}$, the one step estimate contains stronger finite time information that
\[ \frac{\phi(1+ (C-1)e^{-t})}{\phi(C)}\le e^{-2\alpha(C)t}. \]
Moreover, since $\eta_C(1)=1$, the exponential decay rate at $t=0$ is
\begin{align}
-\frac12\left.\frac d{dt}\right|_{t=0}
\ln \eta_C(e^{-t})
&=\frac12\left.\frac d{dp}\right|_{p=1}
\ln \eta_C(p) 
=\frac12\frac{(C-1)\ln C}{\phi(C)}\nonumber\\
&=\frac12\left(1+\frac{C-\ln C-1}{C\ln C-C+1}\right)\, ,
\label{eq:gd-tangent}
\end{align}
which exactly the MLSI lower bound obtained from the reverse-ratio
estimate.  Hence, the reverse ratio result is the infinitesimal counterpart of
the finite time convexity ratio bound.  Notice that this calculation concerns
the upper bound of $\eta(\Phi_p^E)$, not its exact value for a fixed conditional expectation $E$.}
\end{remark}

\subsection{Tensorization}
\label{subsec:gd-tensorization}

Tensorization is an important property of functional inequalities that transfers local estimate to many-body
dynamics without loss in the number of components.  For two independently
evolving Markov semigroups, their product semigroup at time \(t\ge 0\) is
\(P_t^{E_1}\otimes P_t^{E_2}\).
In the classical setting, the ordinary constants obey exact tensorization
\[
\alpha(P_t\otimes Q_t)=\min\{\alpha(P_t),\alpha(Q_t)\},
\qquad
\eta(P_t\otimes Q_t)=\max\{\eta(P_t),\eta(Q_t)\}.
\]
Nevertheless, the quantum MLSI is not stable under tensoring with a passive ancilla or environment.
For the qubit depolarizing channel
$E_2(\rho)=\tr(\rho)\mathds 1_2/2$, one has
\begin{equation}
\frac12
\le \alpha(\id_{M_2}\otimes P_t^{E_2})
<0.965<1=\alpha(P_t^{E_2}).
\label{eq:gd-qubit-nontensorization}
\end{equation}
See \cite[Section~4.3]{BrannanGaoJunge2022}; the example is also discussed
in~\cite{GaoJungeLaRacuenteLi2025}.  It motivates complete MLSI, which
requires uniform control under
all finite-dimensional amplifications.
\begin{equation}
\alpha_{\mathrm c}(P_t^E)
\!:=\inf_{k\ge1}\alpha(P_t^E\otimes\id_{M_k}),
\qquad
\eta_{\mathrm c}(\Phi_p^E)
\!:=\sup_{k\ge1}\eta(\Phi_p^E\otimes\id_{M_k}).
\label{eq:gd-complete-constants}
\end{equation}

The conditional expectation chain rule gives the exact tensorization
\cite{gao2020fisher,GaoRouze2022}
\begin{align}
\alpha_{\mathrm c}(P_t^{E_1}\otimes P_t^{E_2})
&=\min\{\alpha_{\mathrm c}(P_t^{E_1}),
         \alpha_{\mathrm c}(P_t^{E_2})\}, \nonumber\\
\eta_{\mathrm c}(P_t^{E_1}\otimes P_t^{E_2})
&=\max\{\eta_{\mathrm c}(P_t^{E_1}),
         \eta_{\mathrm c}(P_t^{E_2})\}.
\label{eq:gd-complete-decomposition}
\end{align}
In particular, these strong constants satisfies the following system size independent
consequence for every $n\ge 1$,
\[
\alpha_{\mathrm c}
\bigl((P_t^E)^{\otimes n}\bigr)
=
\alpha_{\mathrm c}(P_t^E),
\qquad
\eta_{\mathrm c}
\bigl((P_t^E)^{\otimes n}\bigr)
=
\eta_{\mathrm c}(P_t^E)\,.
\]

For a $d$-dimensional qudit, let
$E_d(\rho)=\tr(\rho)\mathds 1_d/d$ be the completely
depolarizing channel.  Although $n$ qudits live in dimension $d^n$,
the tensorization reduces their local depolarizing
dynamics to the one-qudit constants, with no loss in $n$.  This is the
relevant scaling for noisy quantum circuits and variational algorithms
\cite{FrancaGarciaPatron2021}.

Our condition expectation formulation immediately applies to the complete entropy contraction constants by
applying Corollary \ref{cor:gd-MLSI} and \ref{cor:gd-entropy-contraction} to every amplification and using the complete index
$ C_{\mathrm{cb}}(E)\ge C(E\otimes\id_{M_k})$.

\begin{corollary}[Complete entropy contraction]
\label{cor:gd-complete-index-bounds}Let $E$ be a faithful conditional expectation and $C_{cb}(E)$ be its complete index.
\begin{align}
\alpha_{\mathrm c}(P_t^E)
&\ge
\frac12\left(
1+\frac{C_{\mathrm{cb}}(E)-\ln C_{\mathrm{cb}}(E)-1}
{C_{\mathrm{cb}}(E)\ln C_{\mathrm{cb}}(E)-C_{\mathrm{cb}}(E)+1}
\right),
\label{eq:gd-cMLSI-bound}
\\
\eta_{\mathrm c}(\Phi_p^E)
&\le
\frac{\phi(1+p(C_{\mathrm{cb}}(E)-1))}
     {\phi(C_{\mathrm{cb}}(E))}\,.
\label{eq:gd-complete-entropy-contraction-bound}
\end{align}
\end{corollary}
\subsection{Index witnesses and tightness}
\label{subsec:gd-index-tightness}
The key ingredient for our tightness  analysis is the existence of an index achieving element for a condition expectation. The proposition below is an extension of \cite[]{PimsnerPopa1986}.

\begin{proposition}[Index witnesses]
\label{prop:gd-commuting-index-witness}
Let $E$ be a faithful conditional expectation channel and denote $C=C(E)$.  Then there is
a pure state $\omega=\lvert\psi\rangle\langle\psi\rvert$ such that
\begin{equation}
E(\omega)\lvert\psi\rangle=C^{-1}\lvert\psi\rangle. 
\label{eq:gd-index-witness-eigenvector}
\end{equation}
Hence $[\omega,E(\omega)]=0$, and $C$ is the optimal constant in
$\omega\le CE(\omega)$. 
For all sufficiently large $k$, we have $C(E\otimes \id_{M_k})=C_{\mathrm{cb}}(E)$ and a pure state witness for $C_{cb}(E)$.
\end{proposition}

\begin{proof}
Put $r_i=\min(m_i,n_i)$.  Let
$\tau_i v_{i,j}=s_{i,j}v_{i,j}$ and choose orthonormal
$u_{i,j}\in\mathcal H_i\otimes\mathbb C$ for
$1\le j\le r_i$.  With
\[
C=\sum_i\sum_{j=1}^{r_i}s_{i,j}^{-1},
\qquad
\lvert\psi\rangle
=\sum_i\sum_{j=1}^{r_i}
  \frac{1}{\sqrt{Cs_{i,j}}}\,u_{i,j}\otimes v_{i,j},
\]
the block formula \eqref{eq:gd-Estar-block} gives
$E(\omega)\lvert\psi\rangle=C^{-1}\lvert\psi\rangle$.
\end{proof}
We have the following two-point mechanism for the tightness analysis.
\begin{lemma}[Binary reduction]
\label{lem:gd-general-binary-reduction}
Let $\sigma$ be a state and let
$\omega=\lvert\psi\rangle\langle\psi\rvert$ be a  pure state.  Suppose that for some $C>1$,
\[
\sigma\lvert\psi\rangle=C^{-1}\lvert\psi\rangle\ .
\]
 Set
$\rho_a=(1-a)\sigma+a\omega$ for 
$-(C-1)^{-1}<a<1$ and $x_a= a+(1-a)C^{-1}\in (0,1)$.
Then, for $0<p<1$,
\begin{align}
\frac{D(\sigma\|\rho_a)}{D(\rho_a\|\sigma)}
=
\frac{D_2(C^{-1}\|x_a)}{D_2(x_a\|C^{-1})}, 
\qquad
\frac{D(p\rho_a+(1-p)\sigma\|\sigma)}
     {D(\rho_a\|\sigma)}
=
\frac{D_2(x_{pa}\|C^{-1})}{D_2(x_a\|C^{-1})}.
\label{eq:gd-index-binary-ratios}
\end{align}
where $D_2$ is the binary relative entropy 
\[
 D_2(u\|v):=u\ln\frac uv+(1-u)\ln\frac{1-u}{1-v}\ ,\  0<u,v<1.
\]
\end{lemma}

\begin{proof}
The eigenvector assumption gives the orthogonal decompositions
\[
\sigma=C^{-1}\omega\oplus\sigma_\perp,
\qquad
\rho_a=
\left(a+(1-a)C^{-1}\right)\omega
\oplus(1-a)\sigma_\perp.
\]
The claims follow by evaluating both relative
entropies on these two blocks.
\end{proof}
Given an index from Prop. \ref{prop:gd-commuting-index-witness}, set
$\sigma=E(\omega)$ and
$\rho_{a}=(1-a)\sigma+a\omega$.  Then
\[
E(\rho_{a})=\sigma,
\qquad
\Phi_p^{E}(\rho_{a})=\rho_{pa}.
\]
Lemma~\ref{lem:gd-general-binary-reduction}, with $C=C(E)$, therefore turns
both entropy ratios exactly into the binary expressions
\eqref{eq:gd-index-binary-ratios}.  
For $C>1$ and $0<p<1$, set
\begin{equation}
\begin{aligned}
 \mathsf U(C)
 &:=\frac12\left(1+
 \inf_{\substack{0<x<1\\x\ne C^{-1}}}
 \frac{D_2(C^{-1}\|x)}{D_2(x\|C^{-1})}\right),
 \nonumber\\
 \mathsf L_p(C)
 &:=\sup_{\substack{0<x<1\\x\ne C^{-1}}}
 \frac{D_2(px+(1-p)C^{-1}\|C^{-1})}
      {D_2(x\|C^{-1})}.
\end{aligned}
\label{eq:gd-binary-bound-functions}
\end{equation}
 For fixed
$0<p<1$, the calculation in Appendix~\ref{app:binary-asymptotics} gives
\begin{equation}
 \begin{aligned}
  \mathsf U(C)&=\frac12+\frac{1}{2\ln C}
  +O\!\left(\frac{\ln\ln C}{(\ln C)^2}\right),\\
  \mathsf L_p(C)&=p+\frac{p\ln p}{\ln C}
  +O_p\!\left(\frac{\ln\ln C}{(\ln C)^2}\right),
  \qquad C\to\infty.
 \end{aligned}
 \label{eq:gd-binary-bound-asymptotics}
\end{equation}
The variational definitions and \eqref{eq:gd-index-binary-ratios}, together
with Corollaries~\ref{cor:gd-MLSI} and
\ref{cor:gd-entropy-contraction}, give, with $C=C(E)$,
\begin{equation}
 \begin{aligned}
  &\frac12\left(1+\frac{C-1-\ln C}{C\ln C-C+1}\right)
  \le \alpha(P_t^E)\le \mathsf U(C)
  \\
  &\mathsf L_p(C)\le \eta(\Phi_p^E)
  \le \frac{\phi(1+p(C-1))}{\phi(C)}.
 \end{aligned}
 \label{eq:gd-binary-constant-bounds}
\end{equation}

The two-sided bounds above give the following first asymptotic form.
\begin{theorem}
Let $E$ be a faithful conditional expectation with $C=C(E)$.  For every fixed $0<p<1$, 
\begin{equation}
 \begin{aligned}
  \alpha(P_t^E)
  &=\frac12+\frac{1}{2\ln C}
    +O\!\left(\frac{\ln\ln C}{(\ln C)^2}\right),\\
  \eta(\Phi_p^E)
  &=p+\frac{p\ln p}{\ln C}
    +O_p\!\left(\frac{\ln\ln C}{(\ln C)^2}\right),
 \end{aligned}
 \label{eq:gd-ordinary-constant-asymptotics}
\end{equation}
as $C\to\infty$. In particular, the corresponding conclusion holds uniformly over all finite
dimensional amplifications with $C=C_{\mathrm{cb}}(E)$.
\end{theorem}
\begin{remark}[Relation to the sharp depolarizing result]{\rm 
For the $\sigma$-biased depolarizing semigroup,  $E_\sigma(\rho)=\tr(\rho)\sigma$,
$C(E_\sigma)=\lambda_{\min}(\sigma)^{-1}$.  In this case the binary witness
above coincides with the sharp reduction of M\"uller-Hermes, Stilck
Fran\c{c}a, and Wolf
\cite[Lemma~3.1 and Theorem~3.1]{MullerHermesFrancaWolf2016}.  With
$s=\lambda_{\min}(\sigma)$, their result gives
\begin{equation}
\alpha(P_t^{E_\sigma})
=\frac12\left(1+
\inf_{\substack{0<x<1\\x\ne s}}
\frac{D_2(s\|x)}{D_2(x\|s)}\right)
=\mathsf U(C(E_\sigma)).
\label{eq:gd-known-exact}
\end{equation}
Our
binary construction yields universal bounds for entropy contraction in the same spirit of their reduction. Nevertheless, we make sharpness claim for the 
convexity ratio or reverse ratio.}
\end{remark}
\section{Examples}

\label{sec:examples}
We discuss some examples of conditional expectations and the associated generalized depolarization.
Define the two functions
\begin{equation}
 \mathsf A(x):=\frac12\left(1+
 \frac{x-\ln x-1}{x\ln x-x+1}\right),
 \qquad
 \mathsf B_p(x):=\frac{\phi(1+p(x-1))}{\phi(x)}\,.
 \label{eq:example-bound-functions}
\end{equation}
The values at \(x=1\) are understood by continuity.  The preceding ordinary
and complete bounds give
\begin{equation}
\begin{aligned}
 \alpha(P_t^E)&\ge \mathsf A(C(E)),
 &\eta(\Phi_p^E)&\le \mathsf B_p(C(E)),\\
 \alpha_{\mathrm c}(P_t^E)&\ge \mathsf A(C_{\mathrm{cb}}(E)),
 &\eta_{\mathrm c}(\Phi_p^E)&\le \mathsf B_p(C_{\mathrm{cb}}(E)).
\end{aligned}
\label{eq:example-four-bounds}
\end{equation}
Table~\ref{tab:examples-indices} summarizes the Heisenberg-picture fixed-point
algebras \(\mathcal N=\operatorname{Ran}E^*\) of the examples and their indices.  In the
subsystem-reset row,
\(\lambda_r^\uparrow=\lambda_r^\uparrow(\tau_B)\) denotes the eigenvalues
of \(\tau_B\) in increasing order.
\begin{table}[htbp]
\centering
\caption{Fixed-point algebras and index constants}
\label{tab:examples-indices}
\small
\renewcommand{\arraystretch}{1.3}
\begin{tabularx}{\textwidth}{@{}>{\raggedright\arraybackslash}p{4cm}
  >{\raggedright\arraybackslash}X
  >{\centering\arraybackslash}p{2.75cm}
  >{\centering\arraybackslash}p{2.45cm}@{}}
\toprule
Model & Fixed-point algebra & $C(E)$ & $C_{\mathrm{cb}}(E)$ \\
\midrule
Depolarizing to $\sigma$
  & $\mathbb C\mathds 1$
  & $\lambda_{\min}(\sigma)^{-1}$
  & $\tr(\sigma^{-1})$ \\
\addlinespace
Subsystem reset to $\tau_B$
  & $\mathcal B(\mathcal H_A)\otimes\mathbb C\mathds 1_B$
  & $\sum_{r=1}^{\min(d_A,d_B)}(\lambda_r^\uparrow)^{-1}$
  & $\tr(\tau_B^{-1})$ \\
\addlinespace
$m$-block dephasing
  & $\displaystyle\bigoplus_{k=1}^m\mathcal B(P_k\mathcal H)$
  & $m$
  & $m$ \\
\addlinespace
Compact group twirling
  & $\bigoplus_\alpha
     \mathcal B(M_\alpha)\otimes\mathbb C\mathds 1_{V_\alpha}$
  & $\sum_\alpha d_\alpha\min(m_\alpha,d_\alpha)$
  & $\sum_\alpha d_\alpha^2$ \\
\bottomrule
\end{tabularx}
\end{table}
\subsection{Depolarizing}
The replacer channel $E_\sigma(\rho)=\tr(\rho)\sigma$ is a conditional expectation to the trivial subalgebra $\mathbb{C}1$.
By the formula in Proposition \ref{prop:gd-index-structure}, the indices are
$$C(E_{\sigma})=\|\sigma^{-1}\|_\infty\ ,\  C_{\mathrm{cb}}(E_{\sigma})=\tr(\sigma^{-1}).$$  
\begin{equation}
 \boxed{
 \begin{aligned}
 \alpha(P_t^{E_\sigma})&\ge \mathsf A(\|\sigma^{-1}\|_\infty),
  &\eta(\Phi_p^{E_\sigma})&\le \mathsf B_p(\|\sigma^{-1}\|_\infty),\\
  \alpha_{\mathrm c}(P_t^{E_\sigma})&\ge \mathsf A(\tr(\sigma^{-1})),
  &\eta_{\mathrm c}(\Phi_p^{E_\sigma})&\le \mathsf B_p(\tr(\sigma^{-1})).
 \end{aligned}}
 \label{eq:depolarizing-four-bounds}
\end{equation}
Here we note that the sharp value of $\alpha(P_t^{E_\sigma})$ was discussed in \eqref{eq:gd-known-exact}.

\begin{example}[Thermal reset]
Let \(H=\sum_n h_n|n\rangle\langle n|\) and, for \(\beta\ge0\), let
\(\sigma=e^{-\beta H}/Z_\beta\) be the Gibbs state, where
\(Z_\beta=\sum_n e^{-\beta h_n}\) is normalization constant.  Then
\[
 C(E_\sigma)=Z_\beta e^{\beta h_{\max}},
 \qquad
 C_{\mathrm{cb}}(E_\sigma)
 =Z_\beta\sum_n e^{\beta h_n}.
\]
\end{example}

\begin{example}[Noisy variational circuits]{\rm 
A concrete use of the preceding tensorization bounds arises in variational
circuits \cite{FrancaGarciaPatron2021,YanDuChenMa2025}, where local noise applied after each unitary layer is the standard qubit depolarizing given by
\[ E_2(\rho)=\tr(\rho)\mathds 1_2/2, \qquad
\Phi_p^{E_2}=p\,\id+(1-p)E_2,\] where \(0<p<1\) is the retained
weight.  The relevant system-size-independent statement is the complete
contraction estimate.  Since
\(C_{\mathrm{cb}}(E_2)=4\), our complete index bound gives
\begin{equation}
 p^2=\eta\bigl((\Phi_p^{E_2})^{\otimes n}\bigr)
 \le \eta_{\mathrm c}\bigl((\Phi_p^{E_2})^{\otimes n}\bigr)= \eta_{\mathrm c}(\Phi_p^{E_2})
 \le \mathsf B_p(4)
 =\frac{\phi(1+3p)}{\phi(4)}<p\,.
\label{eq:gd-qubit-circuit-improvement}
\end{equation}
where the non-complete qubit entropy contraction rate $\eta\bigl((\Phi_p^{E_2})^{\otimes n}\bigr)=p^2$ was obtained in \cite{KastoryanoTemme2013,MullerHermesFrancaWolf2016}. Here the advantage of our $\eta_{\mathrm c}$ constant is not only automatic tensor stable but also controls the entropy with environment or reference systems.
}
\end{example}

\begin{remark}{\rm 
In a manuscript in preparation, Hirche, George, Nuradha, and Wilde
\cite{HircheGeorgeNuradhaWilde2026}
show that for generalized depolarizing semigroups with full rank fixed state $\sigma$, the non-complete entropy contraction rates are tensor stable,
\[ 
 \alpha\bigl((P_t^{E_\sigma})^{\otimes n}\bigr)
 =\alpha(P_t^{E_\sigma}),
 \qquad
 \eta\bigl((\Phi_p^{E_\sigma})^{\otimes n}\bigr)
 =\eta(\Phi_p^{E_\sigma}),\ \  \forall \ n\ge 1
\]
Therefore, in the depolarizing case, the non-complete value $\alpha((P_t^{E_\sigma})^{\otimes n})$ and $\eta((\Phi_p^{E_\sigma})^{\otimes n})$ are determined by the single-system constants.
Moreover, they give some criterion on the exact tensorization of entropy contractions. 
Our bounds for complete entropy contraction rate remain meaningful with arbitrary ancilla system. 
}
\end{remark}

\subsection{Dephasing and coherence}
Let $\{P_k\}_{k=1}^m$ be orthogonal projections with
$\sum_k P_k=\mathds 1$.  Block pinching is given by
\[
 E(\rho)=\sum_{k=1}^m P_k\rho P_k.
\]
Its fixed-point algebra is
\(\mathcal N=\bigoplus_{k=1}^m\mathcal B(P_k\mathcal H)\), and the index formula
\eqref{eq:gd-index-formulas} gives
\[
 C(E)=C_{\mathrm{cb}}(E)=m.
\]
The dynamics therefore removes coherence between the blocks while preserving
coherence inside each block, and all four estimates in
\eqref{eq:example-four-bounds} use \(\mathsf A(m)\) and
\(\mathsf B_p(m)\).

For rank-one projections in a chosen basis $\{|k\rangle\}_{k=1}^d$, \(m=d\) and we have the standard pinching map
\[
 \Delta(\rho)=\sum_{k=1}^d|k\rangle\langle k|\rho|k\rangle\langle k|\,,
\]
whose fixed points are the incoherent states.  The relative entropy of
coherence, introduced in~\cite{BaumgratzCramerPlenio2014}, is
\begin{equation}
 C_{\mathrm r}(\rho):=\min_{\delta=\Delta(\delta)}D(\rho\|\delta)
 =D(\rho\|\Delta(\rho))=S(\Delta(\rho))-S(\rho)\,.
 \label{eq:relative-entropy-coherence}
\end{equation}
It is a standard coherence monotone and, under incoherent operations, admits many operational meaning such as
the asymptotic distillable coherence~\cite{WinterYang2016} (see
\cite{StreltsovAdessoPlenio2017} for a review of the resource theory).

Since both \(\Phi_p^\Delta\) and \(P_t^\Delta\) leave
\(\Delta(\rho)\) unchanged, our bounds directly quantify the loss of this
resource:
\begin{equation}
 C_{\mathrm r}(\Phi_p^\Delta(\rho))
 \le \mathsf B_p(d)C_{\mathrm r}(\rho),
 \qquad
 C_{\mathrm r}(P_t^\Delta(\rho))
 \le e^{-2\mathsf A(d)t}C_{\mathrm r}(\rho).
 \label{eq:coherence-decay-bounds}
\end{equation}
For a state \(\rho_{RA}\), with dephasing acting on \(A\) and arbitrary
quantum side information \(R\), the complete version controls
\(D(\rho_{RA}\|(\id_R\otimes\Delta)(\rho_{RA}))\) uniformly in
\(\dim R\).  For higher-rank \(P_k\), the same relative-entropy functional
measures inter-block coherence; related dephasing estimates appear
in~\cite{MunchWirthZhang2024}.

\begin{example}[Hamiltonian dephasing and time-translation symmetry]{\rm 
Let \(\mathcal H=(\mathbb C^2)^{\otimes n}\) and $Z_r$ be the Pauli $Z$-operator on $r$-th qubit. Consider the Hamiltonian
\[
 H=Z_1+\cdots+Z_n,\qquad
 Z_r=\mathds 1^{\otimes(r-1)}\otimes Z\otimes
 \mathds 1^{\otimes(n-r)}\,.
\]
Denote by \[\mathcal H_k=\text{span}\{ |{x_1\cdots x_n}\rangle \ | \ x_i\in \{0,1\}, \sum_{i=1}^n x_i=k \} \] the span of computational basis vectors of
Hamming weight \(k\), and by \(P_k\) its projection.  Then
\[
 \mathcal H=\bigoplus_{k=0}^n\mathcal H_k,\qquad
 H=\sum_{k=0}^n(n-2k)P_k,\qquad
 \dim\mathcal H_k=\binom nk\,.
\]
For the time-translation representation
\(\pi(\theta)=e^{-i\theta H}\), the invariant subalgebra has the concrete form
\begin{equation}
 \mathcal N=\pi(U(1))'=\{H\}'
 =\bigoplus_{k=0}^n\mathcal B(\mathcal H_k)
 =\left\{\sum_{k=0}^n P_k x P_k:x\in\mathcal B(\mathcal H)\right\}.
 \label{eq:time-translation-commutant}
\end{equation}
The corresponding conditional expectation is given by the time average is
\begin{equation}
 E_H(\rho)=\frac1{2\pi}\int_0^{2\pi}
 e^{-i\theta H}\rho e^{i\theta H}\,d\theta
 =\sum_{k=0}^n P_k\rho P_k\,,
 \label{eq:time-translation-twirl}
\end{equation}
which is also the Pinching map onto the commutant $\{H\}'$ of $H$.
It removes coherence between distinct energy sectors while retaining
coherence within each degenerate sector.  The quantity
\[
 D(\rho\|E_H(\rho))=S(E_H(\rho))-S(\rho)
\]
is therefore the relative entropy asymmetry under time translations \cite{Marvian2022}.
Since \(C(E_H)=C_{\mathrm{cb}}(E_H)=n+1\), our MLSI estimate becomes
\begin{equation}
 D\!\left(P_t^{E_H}(\rho)\middle\|E_H(\rho)\right)
 \le e^{-2\mathsf A(n+1)t}D(\rho\|E_H(\rho)),
 \qquad t\ge0\,.
 \label{eq:time-translation-entropy-decay}
\end{equation}
Hence inter sector coherence decays exponentially at a rate controlled only
by the number \(n+1\) of energy sectors; the same estimate holds uniformly
after adjoining an arbitrary reference system. }
\end{example}

\subsection{Compact group asymmetry}
Let \(G\) be a finite or compact group and let
\(\pi:G\to U(\mathcal H)\) be a projective unitary
representation.  The Haar twirling channel is
\[
 E_G(\rho)=\int_G\pi(g)\rho \pi(g)^*\,d\mu(g).
\]
$E_G=E_G^*$ is a trace-preserving conditional expectation.  Indexing the
irreducible sectors occurring in \(\pi\) by \(\alpha\), write the irreducible
decomposition as
\[
 \mathcal H=\bigoplus_\alpha M_\alpha\otimes V_\alpha,
 \qquad
 \pi(g)=\bigoplus_\alpha
 \mathds 1_{M_\alpha}\otimes\nu_\alpha(g),
\]
where \(m_\alpha:=\dim M_\alpha\) is the multiplicity and
\(d_\alpha:=\dim V_\alpha\) is the dimension of the 
irreducible sector.
Schur's lemma gives
\[
 \mathcal N=\pi(G)'
 =\bigoplus_\alpha\mathcal B(M_\alpha)\otimes
   \mathbb C\mathds 1_{V_\alpha},
 \qquad
 E_G(\rho)=\bigoplus_\alpha
 \tr_{V_\alpha}(P_\alpha\rho P_\alpha)
 \otimes\frac{\mathds 1_{V_\alpha}}{d_\alpha}.
\]
The index formulas reduce to
\[
C(E_G)=\sum_\alpha d_\alpha\min(m_\alpha,d_\alpha),
\qquad
C_{\mathrm{cb}}(E_G)=\sum_\alpha d_\alpha^2.
\]
Here \[ D(\rho\|E_G(\rho))=H(E_G(\rho))-H(\rho)\]
is the asymmetry of relative entropy with respect to the $G$-symmetry \cite{MarvianSpekkens2014}.

\begin{example}[Collective \(SU(2)\) rotations and total spin]{\rm 
For \(G=SU(2)\), let \(\pi(u)=u^{\otimes n}\) act on \(n\)
spin-\(\tfrac12\) particles $(\mathbb{C}^2)^{\otimes n}$. The
Clebsch--Gordan decomposition gives
\begin{equation}
 (\mathbb C^2)^{\otimes n}\cong
 \bigoplus_{j=j_{\min}}^{n/2}M_j\otimes V_j,
 \qquad
 j_{\min}=\begin{cases}0,&n\text{ even},\\ \frac12,&n\text{ odd},\end{cases}
 \label{eq:su2-spin-decomposition}
\end{equation}
where $V_j\cong\mathbb C^{2j+1}$ is the spin-$j$ irrep and $M_j$ its
multiplicity space.  Explicitly,
\begin{equation}
 m_j:=\dim M_j
 =\binom{n}{n/2-j}-\binom{n}{n/2-j-1}
 =\frac{2j+1}{n/2+j+1}\binom{n}{n/2-j}\,,
 \label{eq:su2-multiplicity}
\end{equation}
where out-of-range binomial coefficients are zero.  In this basis
\[
 \pi(u)=u^{\otimes n}
 =\bigoplus_j\mathds 1_{M_j}\otimes u^{(j)}.
\]

Schur's lemma gives the collective twirl
\begin{equation}
 E_{SU(2)}(\rho)=\bigoplus_j\tr_{V_j}(P_j\rho P_j)
 \otimes\frac{\mathds 1_{V_j}}{2j+1}\,.
 \label{eq:su2-twirl}
\end{equation}
Consequently,
\[
 C(E_{SU(2)})=\sum_j(2j+1)\min\{m_j,2j+1\},\qquad
 C_{\mathrm{cb}}(E_{SU(2)})=\sum_j(2j+1)^2\,.
\]
In particular, \(C_{\mathrm{cb}}(E_{SU(2)})=\Theta(n^3)\), so
\eqref{eq:example-four-bounds} gives complete contraction estimates governed
by a polynomial index even though \(\dim\mathcal H=2^n\).}
\end{example}
Collective \(SU(2)\) noise above models an ensemble of spins undergoing the same
unknown spatial rotation, as occurs without a shared Cartesian reference
frame or under spatially correlated control noise.  The multiplicity
subsystems are insensitive to this collective action and therefore provide
natural encodings for reference-frame-independent communication and
decoherence-free quantum information processing
\cite{BartlettRudolphSpekkens2007}.

\appendix

\label{sec:appendix}

\section{Asymptotics of the binary relative entropy}
\label{app:binary-asymptotics}
Recall the binary relative entropy function
\[
 D_2(u\|v):=u\ln\frac uv+(1-u)\ln\frac{1-u}{1-v}\ ,\  0<u,v<1.
\]
For $C>1$, set
\[
 \qquad s=C^{-1},\qquad
 x_C=(1+\ln C)^{-1},\qquad
 y_C=px_C+(1-p)s.
\]
As $C\to\infty$, direct expansion of the binary relative entropy gives
\begin{align*}
 D_2(s\|x_C)
 &=\frac1{\ln C}+O\!\left(\frac1{(\ln C)^2}\right),\\
 D_2(x_C\|s)
 &=1-\frac{\ln \ln C+2}{\ln C}
   +O\!\left(\frac{\ln \ln C}{(\ln C)^2}\right),\\
 D_2(y_C\|s)
 &=p+\frac{p\ln p-p\ln \ln C-2p}{\ln C}
   +O_p\!\left(\frac{\ln \ln C}{(\ln C)^2}\right).
\end{align*}
Consequently,
\begin{align*}
 \frac{D_2(s\|x_C)}{D_2(x_C\|s)}
 &=\frac1{\ln C}+O\!\left(\frac{\ln \ln C}{(\ln C)^2}\right),\\
 \frac{D_2(y_C\|s)}{D_2(x_C\|s)}
 &=p+\frac{p\ln p}{\ln C}
   +O_p\!\left(\frac{\ln \ln C}{(\ln C)^2}\right).
\end{align*}
On the other hand, the reverse and convexity comparisons from
Theorem~\ref{thm:relative-entropy-comparisons} yield
\begin{align*}
 \frac12\left(1+\frac{C-1-\ln C}{C\ln C-C+1}\right)
 &\le \mathsf U(C)
 \le \frac12\left(1+
 \frac{D_2(s\|x_C)}{D_2(x_C\|s)}\right),\\
 \frac{D_2(y_C\|s)}{D_2(x_C\|s)}
 &\le \mathsf L_p(C)
 \le \frac{\phi(1+p(C-1))}{\phi(C)}.
\end{align*}
Finally,
\[
 \frac{C-1-\ln C}{C\ln C-C+1}
 =\frac1{\ln C}+O\!\left(\frac1{(\ln C)^2}\right),
 \qquad
 \frac{\phi(1+p(C-1))}{\phi(C)}
 =p+\frac{p\ln p}{\ln C}+O_p\!\left(\frac1{(\ln C)^2}\right).
\]
Combining these two-sided estimates proves
\eqref{eq:gd-binary-bound-asymptotics}.

\section{Multiplicativity of the complete index}
\begin{proposition}
\label{cor:app-cb-multiplicative}
For faithful finite-dimensional conditional expectations $E_1$ and $E_2$,
\[
C_{\mathrm{cb}}(E_1\otimes E_2)
=C_{\mathrm{cb}}(E_1)C_{\mathrm{cb}}(E_2)\,.
\]
\end{proposition}
\begin{proof}
If the block states of $E_1$ and $E_2$ are
$\{\tau_i^{(1)}\}_i$ and $\{\tau_j^{(2)}\}_j$, then those of
$E_1\otimes E_2$ are
$\{\tau_i^{(1)}\otimes\tau_j^{(2)}\}_{i,j}$.  By
\eqref{eq:gd-index-formulas},
\begin{align*}
C_{\mathrm{cb}}(E_1\otimes E_2)
=\sum_{i,j}\tr\!\left[
  (\tau_i^{(1)}\otimes\tau_j^{(2)})^{-1}\right]
=\sum_{i,j}
  \tr[(\tau_i^{(1)})^{-1}]
  \tr[(\tau_j^{(2)})^{-1}]
=C_{\mathrm{cb}}(E_1)C_{\mathrm{cb}}(E_2),
\end{align*}
where we used
$(\tau_i^{(1)}\otimes\tau_j^{(2)})^{-1}
=(\tau_i^{(1)})^{-1}\otimes(\tau_j^{(2)})^{-1}$ and multiplicativity of
the trace on tensor products.
\end{proof}

\bibliographystyle{plain}
\bibliography{entropycontraction}

\end{document}